\documentclass[11pt]{article}
\usepackage{color,graphicx}
\usepackage{amsmath,amssymb}
\usepackage{paralist,xspace}
\usepackage{tikz}
\usetikzlibrary{arrows,decorations.pathreplacing}

\newtheorem{theorem}{Theorem}{}
\newtheorem{lemma}{Lemma}{}
{}
{}
{}
\newtheorem{definition}{Definition}{}
\newenvironment{proof}{\textbf{Proof}}{}

\newcommand{\qed}{\ensuremath{\hfill\Box}}

\newcommand{\PP}{\mathcal{P}}
\newcommand{\NP}{\mathcal{NP}}
\newcommand{\OPT}{\textsc{Opt}}
\newcommand{\opt}{\textsc{Opt}}

\newcommand{\expected}{\mathbb{E}}

\title{A $(2+\epsilon)$-approximation for precedence constrained single machine scheduling with release dates and total weighted completion time objective.}

\author{Ren\'e Sitters\thanks{Vrije Universiteit Amsterdam, \texttt{r.a.sitters@vu.nl} and CWI Amsterdam, \texttt{sitters@cwi.nl}}  \and Liya Yang\thanks{East China University of Science and Technology, \texttt{ly$\_$yang7@163.com}}\thanks{Research done while visiting Vrije Universiteit Amsterdam and supported in part by the China Scholarship Council under grant NO.201606740022} }

\begin{document}
\maketitle

\abstract{We give a $(2+\epsilon)$-approximation algorithm for minimizing total weighted completion time on a single machine under release time and precedence constraints. This settles a recent conjecture made in~\cite{Skutella2016} 
}

\section{Introduction}

We consider the problem of minimizing the total weighted completion time on a single machine under precedence and release time constraints, denoted by $1|r_j,prec|\sum w_jC_j$ in the standard notation from~\cite{GrahamLLR1979}. An instance is given by a set of jobs $J=\{1,2,\dots,n\}$ and for each $j\in J$ an integer processing time $p_j\ge 0$, release time $r_j\ge 0$, and weight $w_j\ge 0$. Further, we are given a partial order $\prec$ on $J$ representing the precedence constraints between the jobs. (The partial order is transitive, i.e., if $h\prec j$ and $j\prec k$ then $h\prec k$.)
A schedule is defined by a start time $S_j\ge 0$ for each $j$ such that no job starts before its release time, i.e., $S_j\ge r_j$, and no two jobs are processed at the same moment, i.e., for any pair $j,k$ either $S_j\ge S_k+p_k$ or $S_k\ge S_j+p_j$, where the latter must hold in case $j\prec k$. 
The cost of a schedule is the weighted sum of job completion times, $\sum_j w_jC_j$, where $C_j=S_j+p_j$, and the goal is to minimize cost.      
We say that an algorithm is an $\alpha$-approximation algorithm ($\alpha\ge 1$) if for any instance the cost of the algorithm's schedule is at most $\alpha$ times the optimal cost.

The special case without release time constraints  ($r_j=0$ for all jobs $j$) has been well studied but the computational complexity is still not completely settled.
Several 2-approximation algorithms are known~\cite{ChekuriMotwani1999},\cite{ChudakHochbaum1999},\cite{HallSSW1997},\cite{MargotQW2003},\cite{Pisaruk2003},\cite{Schulz1996} and Bansal and Khot~\cite{BansalKhot2009} showed that no $(2-\epsilon)$-approximation algorithm exists  assuming a variant of the unique games conjecture is true. It is yet unknown if this lower bound holds under the common assumption $\PP\neq \NP$. For the problem with release dates, a $3$-approximation algorithm was given by  Schulz~\cite{Schulz1996} and Hall et al.~\cite{HallSSW1997}. Schulz and Skutella~\cite{SchulzSkutella1997} gave an $e+\epsilon$-approximation algorithm by sequencing jobs in order of random $\alpha$-points. Recently, Skutella~\cite{Skutella2016} improved the ratio to $\sqrt{e}/(\sqrt{e}-1)<2.542$  and conjectured that a $(2+\epsilon)$-approximation algorithm exists.  See the papers~\cite{ChekuriKhanna2004},\cite{AmbuhlMMS2011}, and~\cite{Skutella2016} for a more detailed overview of approximation results.
Here, we give a positive answer to the conjecture by presenting a polynomial time $(2+\epsilon)$-approximation algorithm for any constant $\epsilon>0$. 

Our algorithm works as follows. First, we show that with loss of a factor $(1+\epsilon)$ we can decompose the problem into subproblems that can be solved independently. The  final schedule is obtained by  placing the schedules for subproblems one after the other. Each subproblem has the property that  all jobs are scheduled in an interval $[L,\beta L]$ where $L>0$ and $\beta>1$ is a constant depending on $\epsilon$ only. This property is exploited to get a $(2+\epsilon)$-approximation for subproblems. The decomposition is done at random but is easy to derandomize.  
For each subproblem we work as follows. We guess the approximate start time  of $O(1/\epsilon)$ jobs and add that information to the LP. Then we solve the LP and apply list scheduling in order of LP-values. For each subproblem we only need a polynomial number  of guesses  and  return the best solution found.

\section{List Scheduling in order of LP-values}\label{sec:LP+LS}

A common technique  for minimizing total weighted completion time in scheduling is to apply list scheduling in an order that is derived from a linear program relaxation. (See for example~\cite{Schulz1996},\cite{HallSSW1997},\cite{PhillipsSW:1998},\cite{SchulzSkutella1997}.)
In list scheduling, all jobs are in a list and are added to the schedule one by one in an order derived from the list. In the presence of release dates, there are two intuitive versions of List Scheduling. The most natural one is to schedule job as early as possible precisely in the order of the list. Note that this may cause the machine to stay idle while jobs (later in the list) are available. 
Schulz~\cite{Schulz1996} showed that for this version, List Scheduling in order of LP-values is a 3-approximation algorithm for $1|r_j,prec|\sum w_jC_j$. Alternatively, one may schedule at any moment that the machine is idle, the job that comes earliest in the list  \emph{among the available jobs}.  It is this latter variant that we use here. Possibly, our approach in this paper can be modified to work with the natural version of List Scheduling but we do not see a proof that is  substantially easier than the proof presented here.
In our algorithm we shall relabel jobs before applying list scheduling, i.e., we assume w.l.o.g that the ordered list is $1,2,\dots,n$. Further, we assume this order is consistent with the precedence constraints, i.e., $j<k$ whenever $j\prec k$. 

Given a (partial) schedule, we say that the machine is \emph{available} at time $t$ if for any job in the schedule either $C_j\le t$ or $S_j\ge t$. We say that a job $j$ is \emph{available} at time $t$ if (i) $r_j\le t$, (ii) job $j$ was not scheduled yet, (iii) all jobs $k$ with $k\prec j$ have been completed.

\bigskip 

\noindent \textsc{Algorithm List Scheduling (LS):}\\
Let the jobs $J=\{1,2,\dots,n\}$ be labeled such that $j<k$ whenever $j\prec k$. At any moment $t$ that the machine is available, start the job with the smallest index $j$ among the available jobs.\\

We assume w.l.og. that for any given instance, release dates are consistent with precedence constraints, i.e., we assume that $r_j\le r_k$ if $j\prec k$. (In Section~\ref{sec:bounded}, we modify release dates but maintain consistency.)

\begin{lemma}\label{lem:LSprop}
If release dates are consistent with precedence constraints, i.e., $r_j\le r_k$ if $j\prec k$, then LS has the following property:\\
If at time $t$ the machine is available and there is is a job $j$ with $r_j\le t$ and job $j$ has not started yet, then LS starts some job $h\le j$ at time $t$.
\end{lemma}   
\begin{proof}
For any job $h\prec j$ it holds that $r_h\le r_j$. Hence, if  $j$ is not available at time $t$ then there must be some job $h\prec j$ that is available  at time $t$.
\qed\end{proof}\bigskip

Linear programs for scheduling problems typically use three types of variables: Time-indexed variables (indicating wether job $j$ is processed at time $t$), linear ordering variables (indicating whether job $j$ precedes job $k$), and completion time variables. An elegant LP-fomulation for our problem was introduced by Queyranne~\cite{Queyranne1993} and is based on completion time variables only. It was later refined by Goemans~\cite{Goemans96},\cite{Goemans97} to handle release times. Altough the number of constraints in linear program (LP) is exponential, they can separated in polynomial time by efficient submodular function minimization~\cite{Goemans96}

\begin{align}
\text{(LP) }\min &\  Z=\sum\limits_{j=1}^n w_jC_j& \nonumber\\[3mm] 
s.t.&\  C_j\le C_k& \text{ all pairs } j\prec k\label{eq:LPconstraint1} \\[3mm] 
& \ \sum\limits_{j\in U} p_jC_j\ge r_{\min}(U)p(U)+\frac{1}{2}p(U)^2&\text{ all } U\subseteq J.\label{eq:LPconstraint2}
\end{align}\bigskip

In constraint~\eqref{eq:LPconstraint2}, $p(U)=\sum_{j\in U}p_j$ and $r_{\min}(U)=\min\{r_j\mid j\in U\}$. For later use, we define $r_{\max}(U)=\max\{r_j\mid j\in U\}$ and (given a solution) define the values $C_{\min}(U)$ and $C_{\max}(U)$ in the obvious way. The linear program (LP) can obviously be strengthened by adding the constraint $C_j\ge r_j+p_j$. For our analysis it is enough to have the following implied inequality.

\begin{lemma}\label{lem:LB1LP}
$C_j\ge r_j+p_j/2$ for any $j\in J$.
\end{lemma}
\begin{proof}
Let $U=\{j\}$. Then the second LP-constraint~\eqref{eq:LPconstraint2} states
\[p_jC_j\ge p_j(r_j+p_j/2).\]
\qed\end{proof}

For any schedule $\sigma$ and $U\subseteq J$ the inequality $p(U)\le C_{\max}^{\sigma}(U)-r_{\min}(U)$ clearly must hold. The following similar but weaker inequality is implied by (LP).
 
\begin{lemma}\label{lem:lowerboundonLP}
$p(U)\le 2C_{\max}(U)-2r_{\min}(U)$ for any $U\subseteq J$.
\end{lemma}
\begin{proof}
From the second LP-constraint we have 
\[p(U)C_{\max}(U)\ge \sum\limits_{j\in U} p_jC_j\ge r_{\min}(U)p(U)+\frac{1}{2}p(U)^2,\]
for any $U\subseteq J$. Dividing both sides by $P(U)$ gives 
$C_{\max}(U)\ge r_{\min}(U)+\frac{1}{2}p(U)$ which is the inequality of the lemma.
\qed\end{proof}\bigskip

\noindent \textsc{Algorithm LP+LS:}
\begin{itemize}
\item[1)] Solve the linear program (LP). Relabel such that $C_1\le ...\le C_n$ and such that $j<k$ if $j\prec k$. 
\item[2)] Run list scheduling (LS) in the order $1,\dots,n$. Let $C^{\sigma}_j$ be the completion time of job $j$ in the final schedule $\sigma$.
\end{itemize}\bigskip


The algorithm as defined above has an unbounded approximation ratio as shown by the following example. Let $p_1=1,r_1=1$, $p_2=M,r_2=0, w_1=M$ and $w_2=0$. For large $M$, the optimal schedule places jobs in the order $1,2$ and has value $2M$. The algorithm however, will schedule job $2$ first since it is the only available job at time $0$ which gives value $(M+1)M$.  If jobs are relatively small compared to their release time ($p_j\le r_j$ for all $j$) then the algorithm is a $2$-approximation as we show below. The proof follows easily from the lemma below. We shall use this lemma again in Section~\ref{sec:bounded}.

\begin{lemma}\label{lem:lemma}
Let $\sigma$ be a schedule returned by algorithm LP+LS. Let $j\in J$ and let $t$ be the smallest value such that the interval $[t,C^{\sigma}_j]$ has no idle time and only contains jobs $h\le j$. Let $U$ be the set of jobs processed in the interval $[t,C^{\sigma}_j]$. Then,
\begin{equation}\label{eq:lem_eq1}
C_j^{\sigma}\le t+2C_j-2r_{\min}(U).
\end{equation}
Further, if no job completes at time $t$ then
\begin{equation}\label{eq:lem_eq2}
C_j^{\sigma}\le 2C_j.
\end{equation}
If some job $k$ completes at time $t$ then 
\begin{equation}\label{eq:lem_eq3}
r_{\min}(U)> s,
\end{equation}
where $s$ is the start time of the job $k$. 
\end{lemma}

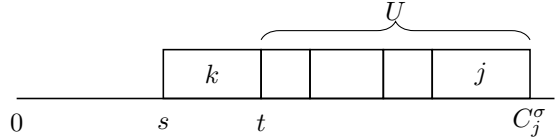
\begin{center}
\begin{tikzpicture}[scale=0.65]
\draw[line width=0.6pt]  (-1,0) rectangle (1,1);
\draw[line width=0.6pt]  (1,0) rectangle (2,1);
\draw[line width=0.6pt] (2,0) rectangle (3.5,1);
\draw[line width=0.6pt] (3.5,0) rectangle (4.5,1);
\draw[line width=0.6pt] (4.5,0) rectangle (6.5,1);

\draw[line width=0.6pt] (-4,0) -- (7,0);
\node at (-4,-0.5) { 0}; 
\node at (-1,-0.5) { $s$ };
\node at (1,-0.5) { $t$};
\node at (5.5,0.5) { $j$}; 
\node at (0,0.5) { $k$};
\node at (6.5,-0.5) { $C_j^{\sigma}$};

\draw[decorate,decoration={brace,amplitude=7pt}] (1,1.2)  -- (6.5,1.2) ; 
\node at (3.75,1.8) { $U$};

\end{tikzpicture}
\end{center}

\noindent\begin{proof}
Note that $U\neq \emptyset$ since $j\in U$. Further, $C_j=C_{\max}(U)$ since only jobs with $h\le j$ are in $U$ and jobs are relabeled in Step 1.  
Now~\eqref{eq:lem_eq1} follows directly from Lemma~\ref{lem:lowerboundonLP}.  
\[
C_j^{\sigma}=t+P(U)\le t+2C_{max}(U)-2r_{\min}(U) = t+2C_j-2r_{\min}(U).
\]
If no job completes at  time $t$ then either $t=0$ or the machine is idle just before time $t$. In the former case, it follows from~\eqref{eq:lem_eq1} that 
$C_j^{\sigma}\le  0+2C_j-2r_{\min}(U)\le 2C_j$. In the latter case it follows from Lemma~\ref{lem:LSprop} that $t=r_{\min}(U)$ which, together with~\eqref{eq:lem_eq1}, implies 
$C_j^{\sigma}\le  2C_j-r_{\min}(U)\le 2C_j$.

Now assume job $k$ completes at time $t$. 
If $r_{\min}(U)\le s$ then, by Lemma~\ref{lem:LSprop} some job $h\le j$ must start at time $s$. However, $k>j$. Hence we must have
$r_{\min}(U)> s$.
\qed\end{proof}

\begin{theorem}\label{lem:2-approx}
If $p_j\le r_j$ for all $j\in J$ then algorithm LP+LS is a 2-approximation. 
\end{theorem}
\begin{proof}
Apply Lemma~\ref{lem:lemma} to an arbitrary job $j$. If no job completes at time $t$ then by~\eqref{eq:lem_eq2} $C_j^{\sigma}\le 2C_j$.
On the other hand, if some job $k$ completes at time $t$ then by~\eqref{eq:lem_eq1} and~\eqref{eq:lem_eq3} 
\[C_j^{\sigma}\le t+2C_j-2r_{\min}(U)< t+2C_j-2s.\]
Since $p_k\le r_k$ we have $t=s+p_k\le s+r_k\le 2s$. Hence, $C_j^{\sigma}\le 2C_j$.
Now take the weighted sum over all jobs: 
\[\sum_j w_jC^{\sigma}_j\le 2\sum_j w_jC_j=2Z_{LP}\le  2\opt.\] 
\qed\end{proof}

Given the theorem above, the following approach leads intuitively to a 2-approximation algorithm. Imagine an unknown optimal schedule $\sigma'$ and for each job $j$ guess if it starts before time $p_j$ in $\sigma'$ and if so, guess its precise start time. Then add this information to the LP and run algorithm LP+LS. Clearly, the running time is not polynomial in general since there can be $O(n)$ of those jobs. 
However, in the next section we show that, with loss of a factor $1+\epsilon$ in the approximation, one can decompose any instance $I$ into subinstance $I_1,I_2,\dots,$ such that our guessing approach is polynomial for each of the subinstances.

\section{A decomposition theorem}\label{sec:decompose}

Margot et al.~\cite{MargotQW2003} and Chekuri and Motwani~\cite{ChekuriMotwani1999} used Sidney's decomposition theorem~\cite{Sidney1975} to show that any instance $I$ of $1 | prec | \sum w_jC_j$  can be split into subinstances $I_1,\dots,I_q$ such that (i) for any subinstance $I_i$  all initial sets $U$ satisfy $p(U)/w(U) \ge  p(I_i)/w(I_i)$ and (ii) 
if $\sigma_1,\dots,\sigma_q$ are $\alpha$-approximate schedules for the subinstance then placing the schedules in this order yields an $\alpha$-approximate schedule for $I$, for any $\alpha\ge 1$. Subsequently, the authors present combinatorial 2-approximation algorithms for instances with this initial set property. 

In the presence of release times, Sidney's decomposition does not hold. Here, we present a different decomposition theorem (Theorem~\ref{th:decomposition}) that is useful in case of release times. The underlying algorithm is given below. Note that we use the common approach of partitioning time into intervals of geometrically increasing length. In general, such a partitioning gives a significant loss in the approximation ratio. We use two techniques to ensure that the (expected) loss is no more than a factor $1+\epsilon$.  First, the factor of increase is taken exponential in $1/\epsilon$. 
This idea was used before in~\cite{Sitters2014} to get a $(1+\epsilon)$-approximation for $1 | prec | \sum w_jC_j$ in case of interval-ordered precedence constraints. Secondly, we use the LP to partition the job set $J$ into subsets $J_i$ where  $J_i$ is the set of jobs that have their $C_j$ in the  $i$-th interval. Using the LP ensures that the schedule for $J_i$ is only a constant factor longer than the length of the $i$-th interval (Lemma~\ref{lem:interval}). This property, together with the large factor of increase ($e^{3/\epsilon}$) and the randomness, ensures that the expected delay due to this partitioning is only a factor $1+\epsilon$.

For the ease of analysis we shall assume that there is no initial set $U$ with $p(U)=0$ since such set can be scheduled at time $0$ and hence can be removed from the instance. What we get from this is that in any LP-solution, $C_j\ge 0.5$ for all $j$ (using Lemma~\ref{lem:LB1LP} and $p_j$ integer).

\bigskip 

\noindent \textsc{Algorithm Decompose:}
\begin{itemize}
\item[1)] Solve the linear program (LP). Let $C_1,\dots,C_n$ be the LP-values. 
\item[2)] Let $a=3/\epsilon$ and take $b$ uniformly at random from $[0,a]$. Let $t_i=e^{a(i-3)+b}$ for $i=1,2,\dots,q$. Choose $q$ large enough such that $C_{max}(J)\le t_q$. Partition the jobs into $J_i=\{j|t_i\le C_j<t_{i+1}\}$, $i\in\{1,2,\dots,q\}$. 
\item[3)] Let $I_i$ be the scheduling instance defined by jobs $J_i$ with the additional constraint that no job  is allowed to start before time $3t_i$. For each $i$, run the algorithm described in Section~\ref{sec:bounded} and let $\sigma_i$ be the schedule returned.
\item[4)] Return $\sigma$ which is the concatenation of $\sigma_1,\dots,\sigma_q$.
\end{itemize}\bigskip
 
First, let us see that the algorithm returns a feasible schedule. Step 1 is the same as in algorithm LP+LS except that relabelling is not needed here. To ensure that $t_1\le C_{\min}(J)$ we used $i-3$ in stead of just $i$ in Step 2.  Note that $t_1=e^{b-2a}\le e^{-a}< 0.5\le C_{\min}(J)$. Hence, Step 2 defines a partition of $J$.
Further, since the partial order $\prec$ is transitive, the instances of Step 3 are well defined. (The precedence constraints between jobs in $J_i$ are the same as in $J$.) Finally, note that the first LP-constraint~\eqref{eq:LPconstraint1} ensures that the precedence constraints are satisfied in $\sigma$ since if $k\prec j$ then $C_k\le C_j$ and  $k$ will be scheduled before $j$ in $\sigma$. Hence, $\sigma$ is feasible if we place the partial schedules in the order $\sigma_1,\dots,\sigma_q$ and shift schedules forward in case of overlap.
However, we will show below that schedule $\sigma_i$ is contained in the interval $[3t_i,3t_{i+1}]$.
That means we can simply take the union of the $\sigma_i's$ and do not need to shift.

\subsection{Analysis} 
Say that a schedule is \emph{tight} if no job can be  shifted to the left (scheduled earlier) while maintaining feasibility and without shifting any of the other jobs. Clearly, a non-tight schedule can be made tight by checking each job. Hence we may assume that the schedules $\sigma_i$ returned by the algorithm described in Section~\ref{sec:bounded} are tight.

\begin{lemma}\label{lem:interval} Any tight schedule for $I_i$ is contained in the interval $[3t_i,3t_{i+1}]$.
\end{lemma}
\begin{proof}
We assumed here that $\epsilon$ is small enough. To be precise, we assume $\epsilon \le 3/\ln 3$ since then $t_{i+1}=e^{3/\epsilon}t_i\ge 3t_i$. 
By definition, no job starts before time $3t_i$. 
Since the schedule is tight the last job completes latest at time 
\[\max\{3t_i,r_{\max}(J_i)\}+P(J_i).\]
Note that $r_{\max}(J_i)\le t_{i+1}$
since $r_j\le C_j\le t_{i+1}$ for all jobs $j\in J_i$. Also, $3t_i\le t_{i+1}$ as we showed above.
Further, by Lemma~\ref{lem:lowerboundonLP}, 
\[P(J_i)\le 2C_{\max}(J_i)\le 2t_{i+1}.\]
Hence, the last job completes latest at time  \[\max\{3t_i,r_{\max}(J_i)\}+P(J_i) \le t_{i+1}+2t_{i+1}=3t_{i+1}.\]
\qed\end{proof}

Let $\OPT_i$ be the optimal value of instance $I_i$. Now, consider an optimal schedule $\sigma^*$ for $I$ and let $\opt|_i$ be the contribution of $J_i$ in the optimal schedule. That means, 
$\opt|_i=\sum_{j\in J_i} w_jC_j^*$, where $C_j^*$ is the completion time of job $j$ in $\sigma^*$. 
 A feasible schedule for $I_i$ is obtained by removing the jobs not in $J_i$ from $\sigma^*$ and shifting the remaining schedule forward  by at most $3t_i$. Hence we get the following bound.
\begin{lemma}\label{lem:shift_opt}
 $\OPT_i\le \opt|_i+3t_i\sum_{j\in J_i}w_j$.
\end{lemma}

The value $\opt_i$ depends on the random variable $b$ which defines the partition. For job $j$ let $i(j)$ be such that $j\in J_{i(j)}$, that means, $t_{i(j)}\le C_j<  t_{i(j)+1}$. Note that $t_{i(j)}$ is a stochastic variable of the form $t_{i(j)}=e^{-x}C_j$ where $x$ is uniform on $[0,a]$. 
\begin{equation}\label{eq:t_i(j)}
\expected[t_{i(j)}] = C_j\expected[e^{-x}]=\frac{C_j}{a}\int\limits_{x=0}^{x=a}e^{-x}dx=\frac{C_j(1-e^{-a})}{a}<\frac{C_j}{a}=\frac{\epsilon C_j}{3}.
\end{equation}

\begin{lemma}\label{lem:EOPTi}
$\expected[\sum_i \OPT_i]\le (1+3/a)\OPT=(1+\epsilon)\OPT$.
\end{lemma}
\begin{proof}
From Lemma~\ref{lem:shift_opt},
\[\sum_i\OPT_i\le \sum_i \opt|_i+3\sum_i\sum_{j\in J_i}w_jt_i= \OPT+3\sum_j w_j t_{i(j)}.\]
From~\eqref{eq:t_i(j)}, the expected value over $b$ is
\[\expected[\sum_i\OPT_i]\le \OPT+3\sum_j w_j \expected[t_{i(j)}]\le \OPT+\frac{3}{a}\sum_j w_j C_j\le (1+\frac{3}{a})\OPT.
\]
\qed\end{proof}

Remember that a schedule is called tight if no job can be shifted left (scheduled earlier) while maintaining feasibility and without shifting any of the other jobs. 

\begin{definition}\label{def:bounded}
We say that an instance of $1|r_j,prec|\sum w_jC_j$ with job set $J$ is \emph{bounded} if there is some number $L>0$ such that  
$r_{\min}(J)\ge L$ and any tight schedule completes within time $\beta L$ for some constant $\beta$. 
\end{definition}

\begin{theorem}\label{th:decomposition}
For any instance $I$ of $1|r_j,prec|\sum w_jC_j$ and constant $\epsilon>0$ we can find bounded instances $I_1,\dots,I_q$ such that if $\sigma_1,\dots,\sigma_q$ are (randomized) $\alpha$-approximate schedules for $I_1,\dots,I_q$ then the schedule obtained by placing the  $\sigma_i$'s in order $i=1,\dots,q$ is a randomized   $\alpha(1+\epsilon)$-approximate schedule for $I$.
\end{theorem}
\begin{proof}
Each $I_i$ is a bounded instance with $L=3t_i$ and $\beta= e^{3/\epsilon}$. If each schedule $\sigma_i$ is an $\alpha$-approximation for instance $I_i$ then the union is feasible (Lemma~\ref{lem:interval}) and has expected value (Lemma~\ref{lem:EOPTi}) at most 
\[\sum_i\expected[\alpha\OPT_i]=\alpha\sum_i\expected[\OPT_i]\le \alpha(1+\epsilon)\OPT.
\]
By Lemma~\ref{lem:interval} there is no overlap in the schedules $\sigma_i$. Hence Note that the union is indeed 
\qed\end{proof}\bigskip

The theorem implies that any (randomized) polynomial time $\alpha$-approximation algorithm  for bounded instances yields a randomized polynomial time  $\alpha(1+\epsilon)$-approximation for general instances.
If the algorithm for the bounded instances is deterministic then we can easily derandomize the combined algorithm by discretizing the probability distribution for $b\in [0,a]$. We show in the next section how to get a deterministic $\alpha$-approximate schedule for bounded instances with $\alpha=2(1+\epsilon)$.

\section{Algorithm for bounded instances.}\label{sec:bounded}
In this section we restrict to bounded instances as defined in Definition~\ref{def:bounded}. Apart from that definition, the analysis here is independent of Section~\ref{sec:decompose}. Hence, let $I$ be any bounded instance with parameters $\beta$ and $L$. 

The main idea of the algorithm is to guess enough information about an (unknown) optimal schedule for $I$ such that the algorithm LP+LS of Section~\ref{sec:LP+LS} yields a $2(1+\epsilon)$-approximate schedule. 
To restrict the number of guesses we first observe (Section~\ref{sec:restrictOPT}) that we only need to consider a nearly optimal schedule $\sigma'$ in which each job $j$ starts at a time that is a multiple of $\epsilon p_j$. Say that a job $j$ is \emph{early} in $\sigma'$ if it starts at time $S'_j<p_j$. 
We will guess the start time of each early job.  The second observation is that the number of early jobs is $O(\log \beta)$ (Lemma~\ref{lem:early}).
With these two observations, the number of guesses is polynomially bounded. For each guess, we adjust release times of jobs in correspondence with our guess and run algorithm LP+LS. The final solution is the best schedule over all guesses. 

\bigskip

\noindent \textsc{Algorithm Bounded:}\\ 
Guess the set $A\subseteq J$ of jobs that are early in the  (near) optimal schedule $\sigma'$ and for each $j\in A$ guess its start time. For each guess, adjust the release times $r_j\rightarrow r'_j$ and run algorithm LP+LS. Let $\sigma$ be the best schedule over all possible guesses.

\subsection{Restricting the optimal schedule.}\label{sec:restrictOPT} Let $\OPT$ be the optimal value for the bounded instance $I$. Consider some (unknown and tight) optimal schedule $\sigma^*$ and let $C^*_1< \dots < C^*_n$ be the completion times. Assume we shift jobs one by one (starting with job 1) such that the start time of each job $j$ is a multiple of $\epsilon p_j$. Let this schedule be $\sigma'$. Then, for any $j$ the new completion time is 
\[C_j'\le C^*_j+\sum_{k\le j}\epsilon p_k=C^*_j+\epsilon \sum_{k\le j} p_k\le (1+\epsilon) C^*_j.\] 
Let $\OPT'$ be the value of $\sigma'$. We have the following properties.
\begin{enumerate}[(i)]
\item $\OPT'\le (1+\epsilon)\OPT$.
\item The start time $S'_j$ of job $j$ is a multiple of $\epsilon p_j$.
\item All jobs are schedule in the interval $[L,(1+\epsilon)\beta L]$.
\end{enumerate}
From now, let $\sigma'$ be our (unknown) near-optimal schedule and let $\OPT'$ be its value. 
We will show how to get a schedule of value at most $2\OPT'$.

\subsection{Guessing the optimal schedule.} \label{subsec:guessing} The first step of the algorithm is to make guesses about $\sigma'$. 
Let $S'_j$ be the start time of job $j$ in $\sigma'$. Say that a job $j$ is \emph{early}  in $\sigma'$ if $S'_j<p_j$.

\begin{lemma}\label{lem:early} The number of early jobs  in $\sigma'$ is $O(\log \beta)$. 
\end{lemma}
\begin{proof}
If $j$ is  an early job in $\sigma'$ then $C'_j >2S'_j$. Hence, the number of early jobs is bounded by 
\[\log_2\left(\frac{(1+\epsilon)\beta L}{L}\right)=\log_2((1+\epsilon)\beta)=O(\log \beta).\]
\qed\end{proof}

Our algorithm  guesses the set $A\subseteq J$ of early jobs in $\sigma'$ and for each early job $j$ we guess its start time $S_j'$  in $\sigma'$. If $j$ is early then there are at most $1/\epsilon$ possibilities to consider since $S'_j<p_j$ and $S'_j$ is a multiple of $\epsilon p_j$. Hence, the total number of guesses is bounded by $(n/\epsilon)^{O(\log \beta)}$.

\subsection{Adjusting release times: $r_j\rightarrow r'_j$}\label{subsec:adjustingLP}
We describe this step under the assumption that our guess about $\sigma'$ is correct, i.e., $A\subseteq J$ is the set of early jobs in $\sigma'$. We increase release times step by step such that the following properties hold. Let $r'_j\ge r_j$ be the new release times and $I'$ the new instance.
\begin{enumerate}[(a)]
\item Schedule $\sigma'$ is feasible for $I'$
\item if $j\in A$ then $r'_j\ge S'_j$  
\item if $j\notin A$ then $r'_j\ge p_j$  
\item if $j\prec k$ then $r'_j\le r'_k$
\end{enumerate}
Let $T$ be the set of \emph{open} time intervals at which an early job is processed, i.e., 
\[T=\mathop{\bigcup}\limits_{j\in A} \left]S'_j,S'_j+p_j\right[.\]
No job starts at a time $t\in T$ in $\sigma'$ so we may increase release times further such that 
\begin{enumerate}[(e)]
\item $r'_j\notin T$ for any $j\in J$
\end{enumerate}
The increase due to (e) may give a conflict with (d) causing a sequence of increases by rules (d) and (e). Clearly, this process ends after a polynomial number of iterations. Note that $|T|$ consists of only $O(\log \beta)$ intervals. A rough upper bound on the number of iterations is $O(n^2\log\beta)$. Hence, we can find release times $r'_j$ such that (a)-(e) hold.

\subsection{The analysis.}

\begin{theorem}\label{lem:2-approx}
Algorithm \textsc{Bounded} returns a $2(1+\epsilon)$ approximate solution for bounded instances.
\end{theorem}
\begin{proof}
Assume that we guessed the information about $\sigma'$ correctly. Let $Z_{LP}$ be the LP value obtained. Then $Z_{LP}\le \OPT'$. 
We will show that $C^{\sigma}_j\le 2C_j$ for any job $j$, where $C_j$ is the optimal LP-value  for $I'$. Then the theorem follows by taking the weighted sum over all jobs: 
\[\sum_j w_jC^{\sigma}_j\le 2\sum_j w_jC_j=2Z_{LP}\le 2\OPT'\le 2(1+\epsilon)\opt.\] 
Consider schedule $\sigma$ returned by algorithm \textsc{Bounded}. Let $j$ be an arbitrary job and apply Lemma~\ref{lem:lemma} to $\sigma$. If no job completes at time $t$ then $C_j^{\sigma}\le 2C_j$. 
Now assume some job $k$ completes at time $t$ and let $s$ be its start time. If $t\le 2s$ then from~\eqref{eq:lem_eq1} and~\eqref{eq:lem_eq3} , 
$C_j^{\sigma}\le t+2C_j-2r'_{\min}(U) < t+2C_j-2s\le 2C_j$. Hence, assume from now that 
\begin{equation}\label{eq:t>2s}
t>2s.
\end{equation}

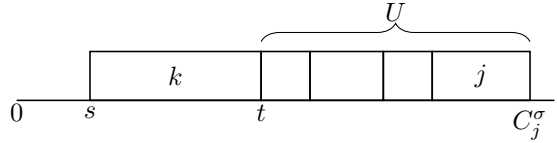
\begin{center}
\begin{tikzpicture}[scale=0.65]
\draw[line width=0.6pt]  (-2.5,0) rectangle (1,1);
\draw[line width=0.6pt]  (1,0) rectangle (2,1);
\draw[line width=0.6pt] (2,0) rectangle (3.5,1);
\draw[line width=0.6pt] (3.5,0) rectangle (4.5,1);
\draw[line width=0.6pt] (4.5,0) rectangle (6.5,1);

\draw[line width=0.6pt] (-4,0) -- (7,0);
\node at (-4,-0.25) { 0}; 
\node at (-2.5,-0.25) { $s$ };
\node at (1,-0.25) { $t$};
\node at (5.5,0.5) { $j$}; 
\node at (-0.75,0.5) { $k$};
\node at (6.5,-0.5) { $C_j^{\sigma}$};

\draw[decorate,decoration={brace,amplitude=7pt}] (1,1.2)  -- (6.5,1.2) ; 
\node at (3.75,1.8) { $U$};
\end{tikzpicture}\\
\end{center}Then, $p_k=t-s>s\ge r_k'$ and by (c) we must have $k\in A$, i.e.,  $k$ is an early job. Since $k\in A$ we have (from (b)) $S'_k\le r_k'\le s$ and   (from~\eqref{eq:lem_eq3})
\begin{equation}\label{eq:rmin>Sk'}
r'_{\min}(U)>s\ge S'_k.
\end{equation}
 Also, since $k\in A$, we know from  (e) that $r'_{\min}(U)\notin\  ]S'_k,S'_k+p_k[$. Together with~\eqref{eq:rmin>Sk'} and~\eqref{eq:t>2s}
we get that
\[r'_{\min}(U)\ge S'_k+p_k\ge p_k=t-s>t/2.\]
Again using~\eqref{eq:lem_eq1} we conclude that 
 \[C_j^{\sigma}\le t+2C_j-2r'_{\min}(U)<2C_j.\] 
\qed\end{proof}
\begin{lemma}\label{lem:polytime}
Algorithm \textsc{Bounded} runs in polynomial time.
\end{lemma}
\begin{proof}
The number of guesses to consider is $(n/\epsilon)^{O(\log \beta)}$. For each guess, adjusting the release times takes $O(n^2/\epsilon)$ time. Also, list scheduling and the LP run in poynomial time, as mentioned earlier in Section~\ref{sec:LP+LS}. The total running for 
algorithm \textsc{Bounded} is $(n/\epsilon)^{O(\log \beta)}$.
 \qed \end{proof}

\section{Conclusion and remarks.}

\paragraph{Total running time.}
The linear program (LP) is solved once to partition an instance $I$ into $O(n)$ instances $I_i$. For each $I_i$ the algorithm takes  $(n/\epsilon)^{O(\log \beta)}$ time  where $\log \beta=O(1/\epsilon)$. Hence, the total running time is $(n/\epsilon)^{O(1/\epsilon)}$.

\paragraph{Reducing the running time.}
We can reduce the total running to $f(\epsilon)p(n)$ for some function $f$ and polynomial $p$ by rounding 
the processing times up to powers of $1+\epsilon$. This reduces the number of guesses needed substantially. Say that job $j$ is of \emph{type} $i$ if its rounded processing $p'_j$ is $(1+\epsilon)^i$. Assume that in the near optimal schedule $\sigma'$ the processing times are rounded.  
Note that in $\sigma'$ there is at most one early job of each type. In stead of guessing all early jobs it is enough to guess which of the types do have an early job. Let $S^{(i)}$ be the smallest start time among the jobs of type $i$ in $\sigma'$.  
Say that type $i$ is early if $S^{(i)}<(1+\epsilon)^i$. Let $B$ be the types that are (guessed) to be early. 
Let $I'$ be the instance for adjusted release times $r'_j$ defined by the following rules.
\begin{enumerate}[(a)]
\item Schedule $\sigma'$ is feasible for $I'$
\item if $i\in B$ and $j$ is of type $i$ then $r'_j\ge S^{(i)}$  
\item if $i\notin B$ and $j$ is of type $i$ then $r'_j\ge p'_j$  
\item if $j\prec k$ then $r'_j\le r'_k$
\item $r'_j\notin T$ for any $j\in J$
\end{enumerate}
Here, $T$ is again the set of \emph{open} time intervals at which an early job is processed, i.e., 
\[T=\mathop{\bigcup}\limits_{i\in B} \left]S^{(i)},S^{(i)}+(1+\epsilon)^i\right[.\]
The analysis is exactly the same except for the bound on the number of guesses. Note that if type $i$ is early then $(1+\epsilon)^i>L$ since no job starts before time $L$. Further, we must have  $(1+\epsilon)^i<(1+\epsilon)^2\beta L$ since no job completes after time $(1+\epsilon)^2\beta L$ in $\sigma'$. Hence, we only need to consider a range of $O(\log_{(1+\epsilon)} \beta)$ values for $i$. The number of guesses is bounded by $(1/\epsilon)^{O(\log_{(1+\epsilon)}\beta)}=(1/\epsilon)^{O((\log \beta)/\epsilon)}$.

\paragraph{No release times.}
The decomposition theorem of Section~\ref{sec:decompose} is still meaningful in the absence of release times. It then states that an $\alpha(1+\epsilon)$-approximation for the problem $1|prec|\sum_j w_jC_j$ follows from  an $\alpha$-approximation for instances of $1|r_j,prec|\sum_j w_jC_j$ in which all jobs have the same release time $r_j=L$ where $L=\Omega(\sum_j p_j)$. Although this does not seem helpful to get below an approximation ratio of $2$ in general, it can be useful for special type of precedence constraints. Moreover, a similar decomposition can be useful for other optimization problems with total completion time objective.   

\bibliographystyle{abbrv}
\bibliography{Bib}

\end{document}